\documentclass[11pt]{article}
\usepackage[utf8]{inputenc}
\usepackage{fullpage}
\usepackage{amssymb,amsmath,amsthm}
\usepackage{epsfig}
\usepackage{url}
\usepackage{mathtools}
\usepackage{breqn}
\usepackage{color}
\usepackage[nocomma]{optidef}
\usepackage[ruled,linesnumbered]{algorithm2e}
\newcommand{\h}{\hspace*{0.2in}}

\newtheorem{thm}{Theorem}%[section]
\newtheorem{cor}[thm]{Corollary}
\newtheorem{lem}[thm]{Lemma}

\theoremstyle{definition}

\newtheorem*{example*}{Example}

\begin{document}
\title{Characterisation of Super-stable Matchings}
\author{Changyong Hu\thanks{Department of Electrical and Computer Engineering, University of Texas at Austin, Austin, Texas 78705, USA. E-mail: \texttt{colinhu9@utexas.edu.}} \and Vijay K. Garg\thanks{Department of Electrical and Computer Engineering, University of Texas at Austin, Austin, Texas 78705, USA. E-mail: \texttt{garg@ece.utexas.edu}.}
%Research supported by the Lend\"ulet program of the Hungarian Academy of Sciences (MTA), under grant number LP2017-19/2017
}
%\date{}
\maketitle

\begin{abstract}
An instance of the super-stable matching problem with incomplete lists and ties is an undirected bipartite graph $G = (A \cup B, E)$, with an adjacency list being a linearly ordered list of ties. Ties are subsets of vertices equally good for a given vertex. An edge $(x,y) \in E \backslash M$ is a blocking edge for a matching $M$ if by getting matched to each other neither of the vertices $x$ and $y$ would become worse off. Thus, there is no disadvantage if the two vertices would like to match up. A matching $M$ is super-stable if there is no blocking edge with respect to $M$. It has previously been shown that super-stable matchings form a distributive lattice \cite{spieker1995set, manlove2002structure} and the number of super-stable matchings can be exponential in the number of vertices. We give two compact representations of size $O(m)$ that can be used to construct all super-stable matchings, where $m$ denotes the number of edges in the graph. The construction of the second representation takes $O(mn)$ time, where $n$ denotes the number of vertices in the graph, and gives an explicit rotation poset similar to the rotation poset in the classical stable marriage problem. We also give a polyhedral characterisation of the set of all super-stable matchings and prove that the super-stable matching polytope is integral, thus solving an open problem stated in the book by Gusfield and Irving \cite{DBLP:books/daglib/0066875}.
\end{abstract}

\section{Introduction}
An instance of the super-stable matching problem with incomplete lists and ties is an undirected bipartite graph $G = (A \cup B, E)$, with an adjacency list being a linearly ordered list of ties. Ties are disjoint and may contain just one vertex. If vertices $b_1$ and $b_2$ are neighbors of vertex $a$ in the graph $G$, then either $(1)$ $a$ strictly prefers $b_1$ to $b_2$, which we denote as $b_1 \succ_a b_2$; or $(2)$ $a$ is indifferent between $b_1$ and $b_2$, which means $b_1$ and $b_2$ are in a tie in $a$'s adjacency list, and denote as $b_1 =_a b_2$; or $(3)$ $a$ strictly prefers $b_2$ to $b_1$. We say $a$ weakly prefers $b_1$ to $b_2$ if either $a$ strictly prefers $b_1$ to $b_2$ or $a$ is indifferent between $b_1$ and $b_2$, which we denote as $b_1 \succeq_a b_2$. A matching $M$ is a set of disjoint edges in the graph $G$. Let $e=(u,v)$ be an edge contained in the matching $M$. Then, we say that vertices $u$ and $v$ are matched in $M$ and write $u = M(v)$ to denote that $u$ is matched to $v$ in $M$. An edge $(x,y) \in E \backslash M$ is a {\em blocking edge} for a matching $M$ if by getting matched to each other neither of the vertices $x$ and $y$ would become worse off, i.e. $x$ is either unmatched or $x$ weakly prefers $y$ to $M(x)$, and $y$ is either unmatched or $y$ weakly prefers $x$ to $M(y)$. We abuse the notation $y \succeq_x M(x)$ for the case that $x$ is unmatched in $M$. A matching is {\em super-stable} if there is no blocking edge with respect to it. 

Super-stable matchings were first investigated by Irving \cite{irving1994stable}, who gave three classes of stable matchings in the case of preference lists with ties, depending on the way of defining a {\em blocking edge} for a matching $M$. In the weakly stable matching problem an edge $e = (x,y)$ is blocking if by getting matched to each other, both $x$ and $y$ would become better off. In the strongly stable matching problem, an edge $e = (x,y)$ is blocking if one of $x$ and $y$ becomes better off and the other would not be worse off.

In this paper we study the problem of characterising the set of all super-stable matchings. The problem was stated in the book by Gusfield and Irving \cite{DBLP:books/daglib/0066875} as one of the 12 open problems. The structure of the set of all stable matchings in the stable marriage problem without ties is well understood in Gusfield and Irving's book \cite{DBLP:books/daglib/0066875}. Recently, Kunysz et al. \cite{kunysz2018algorithm} gave compact representations for the set of all strongly stable matchings and showed that the construction can be done in $O(mn)$ time, where $n$ and $m$ denote the number of vertices and edges in the graph. Scott \cite{scott2005study} investigated the structure of all super-stable matchings by defining an object that he called meta-rotation, which corresponds to one collection of rotations in some arbitrary tie-breaking instance of the original instance and the time complexity of the construction is $O(m^2)$.

We give two compact representations of the set of all super-stable matchings that can be constructed in, respectively, $O(nm^2)$ and $O(mn)$ time.

The first representation of the set of all super-stable matchings consists of $O(m)$ matchings, each of which is a man-optimal stable matching among all super-stable matchings that contains a given edge. We show that computing such matching for each edge can be reduced to computing a man-optimal super-stable matching in a reduced graph by deleting an appropriate subset of edges in graph $G$. The algorithm is described in Section \ref{sec:irreducible}.

Our second representation explicitly constructs rotations, which are differences between consecutive super-stable matchings in a maximal sequence of super-stable matchings starting with a man-optimal super-stable matching and ending with a woman-optimal super-stable matching. Unlike Scott's \cite{scott2005study} meta-rotation, our rotation is the symmetric difference of two super-stable matchings, which could be a cycle or multiple cycles.

Our construction takes $O(mn)$ time, while Scott's \cite{scott2005study} algorithm takes $O(m^2)$ time. We also show how to efficiently construct a partial order among rotations. This poset can be used to solve other problems connected to super-stable matchings such as the enumeration of all super-stable matchings and the maximum weight super-stable matching problem. Fleiner et al. \cite{fleiner2007efficient} solve the weight super-stable matching by reducing it to the 2-SAT problem and the time complexity is $O(mn\log(W))$, where $W$ is the maximum weight among all edges in $G$. By using the rotation poset constructed in this paper, the weighted problem can also be solved in $O(mn\log(W))$ time.

In this paper we also give a polyhedral characterisation for the set of all super-stable matchings and prove that the super-stable matching polytope is integral. This result implies that the maximum weight super-stable matching problem can be solved in polynomial time. Though the complexity of solving LP is usually higher than combinatorial methods, like in \cite{fleiner2007efficient}, this gives an alternative direction to solve the weighted super-stable matching problem. Previously, it has been shown that the stable matching polytope and the strongly stable matching polytope are integral \cite{vate1989linear, rothblum1992characterization, kunysz2018algorithm}, we complete all three cases by proving that the super-stable matching polytope is integral as well.

We also proved a property called self-duality for the super-stable matching polytope, which also holds for the classical stable matching polytope \cite{teo1998geometry} and the strongly stable matching polytope \cite{kunysz2018algorithm}. 

\subsection{Related Works}
Irving \cite{irving1994stable} gave an $O(m)$ algorithm to find a super-stable matching if it exists. Spieker \cite{spieker1995set} showed that super-stable matchings form a distributive lattice. Further properties of super-stable matchings were proved by Manlove in \cite{manlove2002structure}. Scott \cite{scott2005study} introduced the concept called {\em meta-rotation poset} for super-stable matchings and showed the one-to-one correspondence between super-stable matchings and closed subsets of the poset. 

Irving \cite{irving1994stable} and Manlove \cite{manlove2002structure} gave an $O(m^2)$ algorithm to find a strongly stable matching if it exists. Kavitha et al. \cite{kavitha2007strongly} gave an $O(nm)$ algorithm for the strongly stable matching problem. Manlove \cite{manlove2002structure} showed that strongly stable matchings form a distributive lattice. Kunysz et al. \cite{kunysz2016characterisation} gave a characterisation of all strongly stable matchings and later Kunysz \cite{kunysz2018algorithm} gave a polyhedral description for the set of all strongly stable matchings and proved that the strongly stable matching polytope is integral.

For weakly stable matchings, it is not true that all weakly stable matchings of a given instance always have the same size. Weakly stable matching can be easily found by running the deferred-acceptance algorithm while breaking ties in an arbitrary manner. The problem of computing a maximum-size weakly stable matching is NP-hard, which has been proved by Iwama et al. \cite{iwama1999stable}. Thus finding good approximations of the problem becomes very interesting. For the version when ties are allowed on both sides, the currently best approximation factor is $3/2$ \cite{mcdermid20093, paluch2014faster, kiraly2013linear}. For the case when ties only occur on one side, there are a sequence of works pushing the approximation factor lower. Iwama et al. \cite{iwama201425} gave an $25/17$ approximation algorithm. Huang and Kavitha \cite{huang2014improved} improved it to $22/15$. Later Radnai \cite{radnai2014approximation} improved the approximation factor to $41/28$, then Dean et al. \cite{dean2015factor} pushed the approximation factor to $19/13$. Most recent result by Lam and Plaxton \cite{lam20191+} gave the currently best approximation factor of $1+1/e$.

\section{Preliminaries}
In this section we give some definitions and theorems that are useful in the following sections. 

\begin{thm} \cite{irving1994stable, manlove2002structure}
There is an $O(m)$ algorithm to determine a man-optimal super-stable matching of the given instance or report that no such matching exists.
\end{thm}

\begin{thm}\label{thm: super-stable-vertex} \cite{manlove2002structure}
In a given instance of the super-stable matching problem, the same set of vertices are matched in all super-stable matchings.
\end{thm}

\begin{lem} \label{lem:super-stable-diff} \cite{manlove2002structure}
Let $M,N$ be two super-stable matchings in a given super-stable matching instance. Suppose that, for any agent $p$, $(p, q) \in M$ and $(p, q') \in N$, where $p$ is indifferent between $q$ and $q'$, then $q = q'$. 
\end{lem}

We recall some standard notations and definitions from the theory of matchings under preferences.
For a given edge $(m,w)$, any matching containing $(m,w)$ is called an $(m,w)$-matching. Let us denote the set of all super-stable matchings of $G$ by $\mathcal{M}_G$. Let $\mathcal{M}_G(m,w)$ be the set of all super-stable $(m,w)$-matchings in $G$.

For two super-stable matchings $M$ and $N$, we say that $M$ {\em dominates} $N$ and write $M \succeq N$ if each man $m$ weakly prefers $M(m)$ to $N(m)$. If $M$ dominates $N$ and there exists a man $m$ who prefers $M(m)$ to $N(m)$, then we say $M$ {\em strictly dominates} $N$, write $M \succ N$ and we call $N$ a {\em successor} of $M$. Note that by Lemma \ref{lem:super-stable-diff}, $M \succeq N$ implies $M \succ N$, assuming $M$ is not equal to $N$.

\section{Irreducible Super-stable Matchings} \label{sec:irreducible}
In this section, we give our first representation via irreducible matchings. Birkhoff's representation theorem \cite{birkhoff1937rings} for distributive lattices states that the elements of any finite distributive lattice can be represented as finite sets in such a way that the lattice operations correspond to unions and intersections of sets. The theorem gives a one-to-one correspondence between distributive lattices and partial orders. Our goal is to find the partial order that represents the set of all super-stable matchings.

Distributive lattice is closely related to rings of sets, which is a family of sets that is closed under set unions and set intersections. If the sets in a ring of sets are ordered by set inclusion, they form a distributive lattice. Theory regarding rings of sets and its application to representations of the set of stable matchings in the classical stable marriage problem is well studied by Irving and Gusfield \cite{DBLP:books/daglib/0066875}. Below we give a brief summary of this theory that serves as a preliminary for our algorithm.

Given a finite set $B$, the {\em base} set, a family $\mathcal{F} = \{F_0, F_1, \cdots, F_k\}$ of subsets of $B$ is called a ring of sets over $B$ if $\mathcal{F}$ is closed under set union and intersection. A ring of sets contains a unique minimal element and a unique maximal element. 

For any element $a \in B$, we denote $\mathcal{F}(a)$ the set of all elements of $\mathcal{F}$ that contains $a$. It is obvious that $\mathcal{F}(a)$ is also a ring of sets over $B$. We define $F(a)$ to be the unique minimal element of $\mathcal{F}(a)$. An element $F \in \mathcal{F}$ that is $F(a)$ for some $a \in B$ is called {\em irreducible}. We denote $I(\mathcal{F})$ the set of all irreducible elements of $\mathcal{F}$. We view $(I(\mathcal{F}), \leq)$ as a partial order under the relation $\leq$ of set containment. We give the Birkhoff's representation theorem in the language of rings of sets below.

\begin{thm} \cite{DBLP:books/daglib/0066875}
i) There is a one-to-one correspondence between the closed subsets of $I(\mathcal{F})$ and the elements of $\mathcal{F}$.\\
ii) If $S$ and $S'$ are closed subsets of $I(\mathcal{F})$ that generate $F = \bigcup S$ and $F' = \bigcup S'$ respectively, then $F \subseteq F'$ if and only if $S \subseteq S'$.
\end{thm}

In the context of super-stable matchings, the base set $B$ corresponds to the set of all acceptable pairs $(m,w) \in E$. We define the $P$-set of a super-stable matching $M$ to be the set of all pairs $(m,w)$, where $w$ is either $M(m)$ or a woman whom $m$ weakly prefers to $M(m)$, which corresponds to an element in $\mathcal{F}$. It is obvious that the unique minimal (man-optimal) super-stable matching in $\mathcal{M}_G(m,w)$, if nonempty, is {\em irreducible}.

We describe an $O(|E|)$ algorithm for computing a man-optimal super-stable $(m,w)$-matching in $G$. Algorithm \ref{alg:man-super-stable} essentially constructs a reduced graph $G' \subseteq G$ by removing some edges from $G$ (line \ref{line:remove_start} to line \ref{line:remove_end} in Algorithm \ref{alg:man-super-stable}). After that, the algorithm computes a man-optimal super-stable matching $M'$ in the reduced graph $G'$. By adding back the edge $(m,w)$, the new matching $M \cup (m,w)$ is super-stable in $G$.

\begin{algorithm}
{\bf Input:} the graph $G = (A \cup B, E)$ and preference lists of $G$ and an edge $(m,w) \in E$.\\
{\bf Output:} man-optimal super-stable $(m,w)$-matching or deciding that no such matching exists.\\
$G' \leftarrow G \backslash \{m, w\}$ $//$ remove $m$ and $w$ and all edges that are incident to them \\
{\bf for} $m'$ s.t. $(m', w) \in E$ and $m \preceq_w m'$ {\bf do} \\ \label{line:remove_start}
\h {\bf for} $w'$ s.t. $(m',w') \in E$ and $w \succeq_{m'} w'$ {\bf do}\\
\h\h $G' \leftarrow G' \backslash (m',w')$ \\
\h {\bf end for}\\
{\bf end for}\\
{\bf for} $w'$ s.t. $(m, w') \in E$ and $w \preceq_m w'$ {\bf do}\\
\h {\bf for} $m'$ s.t. $(m',w') \in E$ and $m \succeq_{w'} m'$ {\bf do}\\
\h\h $G' \leftarrow G' \backslash (m',w')$\\
\h {\bf end for}\\
{\bf end for}\\ \label{line:remove_end}
compute man-optimal super-stable matching in $G'$.\\
{\bf if} exists man-optimal super-stable matching $M$ in $G'$ and $M \cup (m,w)$ is super-stable in $G$\\
\h {\bf return} $M \cup (m,w)$\\
{\bf else}\\
\h {\bf return} no super-stable $(m,w)$-matching exists.\\
{\bf end if}
\caption{Computing man-optimal super-stable $(m,w)$-matching}
\label{alg:man-super-stable}
\end{algorithm}

\begin{lem}\label{lem:super-stable-necc}
Let $M$ be a super-stable $(m,w)$-matching. Then $M' = M \backslash \{(m,w)\}$ is a super-stable matching in the reduced graph $G'$.
\end{lem}
\begin{proof}
We need to prove $M' \subseteq G'$ or equivalently none of edges removed from $G$ is matched in $M'$. Suppose not, an edge $(m', w')$ was removed from $G$ and is matched in $M'$. Note that $m' \neq m$ and $w' \neq w$. Hence, it follows that there is an edge $(m, w')$ or $(m', w)$ which caused the removal of $(m', w')$. W.l.o.g, let's assume it is $(m,w')$ which caused the removal of $(m',w')$. Then we have $w \preceq_m w'$ and $m \succeq_{w'} m'$. Obviously, $(m, w')$ is a blocking pair, which leads to a contradiction of $M$ being super-stable.

To prove super-stability of $M'$ is easy. If there were an edge $e$ blocking $M'$, it would also block $M$.
\end{proof}

\begin{lem}\label{lem:super-stable-suff}
Let $M'$ be some super-stable matching in the reduced graph $G'$ if exists. If $M' \cup (m,w)$ is a super-stable matching in $G$, then for each super-stable matching $N'$ in $G'$, $N' \cup (m,w)$ is a super-stable matching in $G$. If $G'$ does not have any super-stable matching, then there is no super-stable $(m,w)$-matching. 
\end{lem}
\begin{proof}
Let $M = M' \cup (m,w)$. Since $M'$ is super-stable in $G'$. It follows that only the removed edges in $E\backslash E'$ can potentially block $M$. We have two cases.\\Case 1. Any edge that is incident to $m$ or $w$ cannot block $M$. W.l.o.g, Suppose that for some $w'$ that is incident to $m$, and $(m,w')$ blocks $M$. Then we have $w' \succeq_m w$. By the construction of $G'$, any edge $(m',w')$ such that $m \succeq_{w'} m'$ was removed. Hence $w'$ must be unmatched in $M$. From Theorem \ref{thm: super-stable-vertex}, $w'$ is unmatched in any super-stable matching of $G$. Let us assume there exists some super-stable $(m,w)$-matching $N$. Then $N' = N \backslash (m,w)$ is super-stable in $G'$. Since $w'$ is unmatched in $N$, $(m, w')$ blocks $N$, contradiction. \\Case 2. Any edge $(m', w')$ such that $m' \neq m$ and $w' \neq w$ cannot block $M$. By the construction of the reduced graph $G'$, the removal of $(m',w')$ was caused by some edge $(m, w')$ or $(m', w)$. W.l.o.g, some edge $(m, w')$ caused the removal of $(m', w')$. Hence, if $w'$ is matched in $M$, then $M(w') \succeq_{w'} m'$. $(m',w')$ does not block $M$. In the case that $w'$ is unmatched in $M$, $w'$ is unmatched in any super-stable matching in $G$. Similar to Case 1, if there exists some super-stable $(m,w)$-matching $N$, then $(m,w')$ blocks N, contradiction. By the same argument, if $M$ is super-stable in $G$, for any other super-stable matching $N'$ in $G'$, $M'$ and $N'$ match the same set of vertices. No edges in $E\backslash E'$ can block $N' \cup (m,w)$.
\end{proof}

\begin{thm}
Let $(m,w)$ be an edge in $G$. There is an $O(m)$ algorithm for computing a man-optimal super-stable $(m,w)$-matching or deciding that no super-stable $(m,w)$-matching exists.
\end{thm}
\begin{proof}
Lemma \ref{lem:super-stable-suff} makes sure if Algorithm \ref{alg:man-super-stable} outputs a matching $M$, then $M$ is super-stable in $G$. Lemma \ref{lem:super-stable-necc} guarantees that if there exists any super-stable matching in $G$, then Algorithm \ref{alg:man-super-stable} would never miss it.
\end{proof}

\begin{thm}
$(I(\mathcal{M}_G), \leq)$ can be constructed in $O(nm^2)$ time.
\end{thm}
\begin{proof}
$I(\mathcal{M}_G)$ can be computed in $O(m^2)$ time by running Algorithm \ref{alg:man-super-stable} for each edge $(m,w) \in E$. The set $I(\mathcal{M}_G)$ has at most $m$ elements. By checking each pair of $I(\mathcal{M}_G)$, we can construct the partial order. Each check takes $O(n)$ time. Thus, the total time is $O(nm^2)$.
\end{proof}
\section{A Maximal Sequence of Super-stable Matchings} \label{sec:sequence}
Representation via irreducible matchings is intuitive, but the time complexity is high. In this section, we give another representation via rotation poset and the time complexity to construct this rotation poset is only $O(mn)$.

Rotation poset derives from the concept of {\em minimal differences} of a ring of sets. A chain $C = \{C_1, \cdots, C_q\}$ in $\mathcal{F}$ is an ordered set of elements of $\mathcal{F}$ such that $C_i$ is an immediate predecessor of $C_{i+1}$ for each $i \in [q]$. The maximal chain is a chain that begins at the minimal element of $\mathcal{F}$, $F_0$ and ends at the maximal element of $\mathcal{F}$, $F_z$. Let $F_i$ and $F_{i+1}$ be two elements of $\mathcal{F}$ such that $F_i$ is an immediate predecessor of $F_{i+1}$. The difference $D = F_{i+1} \backslash F_{i}$ is called a {\em minimal difference} of $\mathcal{F}$. Note that for each two consecutive elements of a chain $C$, there is a minimal difference $D$, we say that $C$ contains $D$. The following two theorems give another version of Birkhoff's representation theorem in the language of minimal differences. The reader can find more details in Irving and Gusfield's book \cite{DBLP:books/daglib/0066875}.

\begin{thm} \cite{DBLP:books/daglib/0066875}
If $F$ and $F'$ are two elements in $\mathcal{F}$ such that $F \subset F'$, then every chain from $F$ to $F'$ in $\mathcal{F}$ contains exactly the same set of minimal differences (in a different order).
\end{thm}

\begin{thm} \cite{DBLP:books/daglib/0066875}
Let $D(\mathcal{F})$ denote the set of all minimal differences in $\mathcal{F}$. For two minimal differences $D$ and $D'$, $D \prec D'$ if and only if $D$ appears before $D'$ on every maximal chain in $\mathcal{F}$. There is a one-to-one correspondence between the elements of $\mathcal{F}$ and the closed subsets of $D(\mathcal{F})$. 
\end{thm}

In the context of super-stable matchings, we want to compute a maximal sequence of super-stable matchings in $\mathcal{M}(G)$, i.e. a sequence $M_0 \succ M_1 \succ \cdots \succ M_z$ where $M_0$ is the man-optimal super-stable matching and $M_z$ is the woman-optimal super-stable matching and for each $1 \leq i \leq z$, there is no super-stable matching $M'$ such that $M_{i-1} \succ M' \succ M_i$. We call a matching $M'$ a strict successor of a matching $M$ if $M'$ is a successor of $M$, i.e. $M \succ M'$ and there exists no super-stable matching $M''$ such that $M \succ M'' \succ M'$. We can solve this problem by computing a strict successor of any super-stable matching $M$.

Let $M$ be a super-stable matching in $G$ and $m$ a vertex in $A$. Suppose that there exists a super-stable matching $M'$ such that $m$ gets a worse partner in $M'$ than in $M$, i.e. $M(m) \succ_m M'(m)$. Let $w = M'(m)$, by Lemma \ref{lem:super-stable-diff}, $w$ must be matched in $M$ and $m \succ_w M(w)$. Hence we are essentially searching for some vertex $w$ such that $M(m) \succ_m w$ and $m \succ_w M(w)$. In Algorithm \ref{alg:seq-super-stable}, the set $E_c$ contains for each man $m$ highest ranked edges incident to him that satisfies the condition above. For each man $m$, the candidate edge $(m,w)$ is not unique, there might be other edge $(m,w')$ that forms a tie with $(m,w)$. While in the case of strict preference list, the candidate edge is unique. 

A strongly connected component $S$ of a directed graph $G$ is a subgraph $S$ that is strongly connected, i.e. there is a path in $S$ in each direction between each pair of vertices of $S$, and is maximal with this property: no additional edges or vertices from $G$ can be included in the subgraph without breaking its property of being strongly connected. We say that $e = (m,w)$ is an outgoing edge of $S$ if $m \in S$ and $w \notin S$.

In Algorithm \ref{alg:seq-super-stable} given below we maintain a directed graph $G_d = (V, E_d)$, whose every edge $(m,w) \in E_d \cap M$ is directed from $w$ to $m$ and every other edge $(m,w)$ is directed from $m$ to $w$. $G_d$ is a subgraph of $G$ that contains the edges the algorithm traverses so far. The basic idea of this algorithm is that for each man $m$ such that $M(m) \neq M_z(m)$, we traverse the preference list of $m$ until we find some candidate edges defined above. We add the edges traversed into $G_d$ and the candidate edges into $G_c$. For each strongly connected component $S$ of $G_d$ without outgoing edges, we try to find a perfect matching on $S$ in $G_c = (V, E_c)$. If we are successful, we find a strict successor of $M$. Otherwise, we modify $G_c$ and $G_d$ by allowing edges of lower ranks. 

\begin{algorithm}
\caption{Computing a maximal sequence of super-stable matchings}
\label{alg:seq-super-stable}
let $M_0$ be the (unique) man-optimal super-stable matching of $G$.\\
let $M_z$ be the (unique) woman-optimal super-stable matching of $G$.\\
$M \leftarrow M_0$\\
let $M'$ contain edge $(m, M(m))$ for each man $m$ such that $M(m) =_m M_z(m)$\\
let $E_d$ contain all edges of $M$\\
let $G_d$ be the directed graph $(V, E_d)$ such that each edge $(m,w) \in E_d \cap M$ is directed from $w$ to $m$ and every other edge $(m,w)$ is directed from $m$ to $w$\\
$E' \leftarrow E\backslash E_d$\\
let $E_c = M'$ and $G_c = (V, E_c)$\\
for each $(m,w) \in M$ remove from $E'$ each edge $(m',w)$ such that $m' \prec_w m$ and each edge $(m,w')$ such that $w' \succeq_m w$\\ \label{line:remove_before}
{\bf repeat}\\
\h{\bf while} $(\exists m \in A)$ $deg_{G_c}(m) = 0$ and $outdeg(S(m)) = 0$ {\bf do}\\
\h\h add the set $E_m$ of top choices of $m$ from $E'$ to $E_d$\\
\h\h {\bf if} $outdeg(S(m)) = 0$ {\bf then}\\
\h\h\h add every edge $(m,w) \in E_m$ such that $m \succ_w M(w)$ and $M(m) \succ_m w$ to $E_c$\\
\h\h\h for each edge $(m,w)$ of $E_c$ that becomes strictly dominated by some added \\
\h\h\h edge $(m', w)$ remove it from $G_c$\\ \label{line:remove-edge-1}
\h\h\h remove $E_m$ from $E'$\\ 
\h\h {\bf end if}\\
\h {\bf end while}\\

\h{\bf for} each $m \in A$ such that $outdeg(S(m)) = 0$ {\bf do}\\
\h\h delete all lowest ranked edge in $E_c \cup E'$ incident to any $w \in S$ such that $w$ is \\ 
\h\h multiple engaged\\ \label{line:remove-edge-2}
\h{\bf end for}\\

\h{\bf while} $(\exists S)$ $outdeg(S) = 0$ and $E_c$ is a perfect matching on $S$ {\bf do}\\ \label{line:terminate}
\h\h$M \leftarrow (E_c\cap S) \cup (M \backslash S)$\\
\h\h$M_i \leftarrow M$\\
\h\h output $M_i$\\
\h\h$i \leftarrow i+1$\\
\h\h update $G_c$ and $G_d$: $E_c\cap S$ contains only edges $(m, M(m))$ such that \\\label{line:update}
\h\h $M(m) =_m M_z(m)$; an edge $(m,w) \in S$ stays in $G_d$ only if $w = M(m)$\\
\h\h and $rank_w(m) \le rank_M(w)$\\
\h{\bf end while}\\
{\bf until} $(\forall m\in A)$ $rank_M(m) = rank_{M_z}(m)$
\end{algorithm}

\subsection{Correctness of Algorithm \ref{alg:seq-super-stable}}
Lemma \ref{lem:update} proves that any edge removed from $G_d$ (line \ref{line:remove_before} and line \ref{line:update}) never block any super-stable matching that the algorithm will output.
\begin{lem} \label{lem:update}
Let $M$ be a super-stable matching in $G$. For any successor $N$ of $M$ such that $N$ is also a super-stable matching in $G$ and each $(m,w) \in M$, any edge $(m, w')$ such that $w' \succeq_m w$ or $(m', w)$ such that $m \succ_w m'$ cannot block $N$.
\end{lem}
\begin{proof}
For any edge $(m,w')$ such that $w' \succeq_m w$, this edge $(m,w')$ cannot block $M$ since $M$ is super-stable. Thus we must have $m \prec_{w'} M(w')$. The matching $N$ is a successor of $M$ from the man's point of view, hence from the woman's point of view, $M$ is a successor of $N$. Then we have $m \prec_{w'} N(w')$ since $M(w') \preceq_{w'} N(w')$, which implies that the edge $(m,w')$ would not block $N$. Similarly, for any edge $(m', w)$ such that $m \succ_w m'$, we have $N(w) \succ_w m'$ since $N(w) \succeq_w M(w)$, which implies that the edge $(m', w)$ would not block $N$.
\end{proof}

\begin{lem} \label{lem:delete_dominate}
No edge deleted in line \ref{line:remove-edge-1} can belong to any super-stable matching $N$ dominated by $M$.
\end{lem}
\begin{proof}
 Suppose that the algorithm deletes an edge $(m,w)$ from $E_c$ because it is dominated by some edge $(m',w)$, i.e. $w$ strictly prefer $m'$ to $m$. We want to show that $(m,w)$ cannot belong to any super-stable matching dominated by $M$. Suppose, for a contradiction, that the edge $(m,w)$ belongs to a super-stable matching $N$ dominated by $M$. $m'$ must match to another woman $w'$ in $N$ and $m'$ strictly prefers $w'$ to $w$, otherwise the edge $(m', w)$ would block $N$. Thus we have that $rank_M(m') \leq rank_{m'}(w') < rank_{m'}(w)$, where the first inequality comes from the fact that $N$ is dominated by $M$. Hence, we have two cases shown as below:\\
 Case $1$: $rank_M(m') < rank_{m'}(w') < rank_{m'}(w)$. In this case, $w'$ must match to a different man $M(w')$ other than $m'$ in $M$. Again because $M$ dominates $N$, $m' \succeq_{w'} M(w')$. If $m' \succ_{w'} M(w')$, by our algorithm, in order to let $(m', w')$ belong to $E_d$, $(m',w')$ must also belong to $E_c$ and this requires that $rank_{m'}(w) = rank_{m'}(w')$, which contradicts with the fact that is $rank_{m'}(w') < rank_{m'}(w)$.\\
 Case $2$: $rank_M(m') = rank_{m'}(w') < rank_{m'}(w)$. First, we rule out the case that $m'$ is indifferent between $M(m')$ and $w' = N(m')$ by Lemma \ref{lem:super-stable-diff}. Thus we must have $w' = M(m') = N(m')$. In order to prove a contradiction, we need the property of strongly connected component. $m'$ and $w$ are in the same strongly connected component, hence there must be a directed path $P$ from $w$ to $m'$. Arc $(w',m')$ is the unique arc that points to $m'$ since $(m',w') \in M$. Hence, there must be an arc $(m'',w')$ in $P$. If $rank_M(m'') < rank_N(m'')$, then $rank_{m''}(w') \leq rank_N(m'')$
 , thus $(m'', w')$ would block $N$. So $M(m'') = N(m'')$. Let $w'' = M(m'')$. $(m'',w'')$ is also in path $P$, Let us continue this process until it reaches to $w$ and we will have $M(w) = N(w)$, which is a contradiction.
\end{proof}

\begin{lem} \label{lem:delete_equal}
No edge deleted in line \ref{line:remove-edge-2} can belong to any super-stable matching $N$ dominated by $M$.
\end{lem}
\begin{proof}
Suppose that an edge $(m,w)$ is deleted in line \ref{line:remove-edge-2}, there must be an edge $(m',w) \in E_c$ such that $m =_w m'$. Let $N$ be a super-stable matching dominated by $M$ that includes the edge $(m,w)$. $m'$ must match to a woman $w'$ in $N$ and $m'$ strictly prefer $w'$ to $w$, otherwise the edge $(m',w)$ would block $N$. $m'$ and $w$ are in the same strongly connected component. By the same argument as in Lemma \ref{lem:delete_dominate}, we will have a contradiction. We omit the proof here.
\end{proof}

\begin{lem}\label{lem:strict}
The output matching $M_i$ is super-stable and a strict successor of $M_{i-1}$.
\end{lem}
\begin{proof}
Note that the algorithm outputs $M_i$ when the edge set $E_c$ is a perfect matching in a strongly connected component $S$ with no outgoing edges and $M_i = (M_{i-1}\backslash S) \cup (E_c \cap S)$. Suppose, for a contradiction, that $M_i$ is blocked by some edge $(m,w) \in E_d$. There are four cases. \\Case 1: $m \notin S$ and $w \notin S$, it is obvious that $(m,w)$ cannot block $M_i$, since it would block $M_{i-1}$ as well. \\Case 2: $m \in S$ and $w \notin S$, this is not possible, because this will imply $S$ has an outgoing edge in $E_d$. \\Case 3: $m \notin S$ and $w \in S$, then $M_i(m)(=M_{i-1}(m)) \succ_m w$, hence $(m,w)$ would not block $M_i$. \\Case 4: $m \in S$ and $w \in S$, if $(m,w)$ never belong to $E_c$, then $M_i(w) \succ_w M_{i-1}(w) =_w m$, $(m,w)$ can not block $M_i$; if $(m,w)$ once belongs to $E_c$ and got deleted later, then $w$ always get a strictly better partner than $m$. We prove that no edge from $E_d$ can block $M_i$. There might be some other edges $e \not\in E_d$ that can potentially block $M_i$. These edges are deleted during the updating of $E_d$. Lemma \ref{lem:update} gives a proof that these set of edges cannot block any matching $N$ that is dominated by $M_{i-1}$. Hence $M_i$ is super-stable.

Next we prove that $M_i$ is a strict successor of $M_{i-1}$. Suppose not and let $m$ be any man in $S$ and $N$ a successor of $M_{i-1}$ such that $M_{i-1}(m) \succ N(m) \succ M_i(m)$. If $(m, N(m)) \in E_c$ and is not deleted during the algorithm, then $(m, M_i(m))$ would not be in $E_c$, which is not true. Since $N$ is a successor of $M$ and is super-stable, by Lemma \ref{lem:delete_dominate} and Lemma \ref{lem:delete_equal}, the edge $(m, N(m))$ can never once belong to $E_c$. Let $w = N(m)$, by our updating rule of $E_d$, we have $N(w) \succeq_w M(w)$. While if $N(w) \succ_w M(w)$, then the edge $(m, w)$ must once belong to $E_c$. Thus we have $N(w) =_w M(w)$, which violates Lemma \ref{lem:super-stable-diff}. 
\end{proof}

\begin{lem}\label{lem:necc}
If $M_i \neq M_z$, the algorithm always outputs a matching.
\end{lem}
\begin{proof}
The algorithm will end without outputting any matching if and only if in line \ref{line:terminate} the while loop, it cannot find any strongly connected component with no outgoing edges. Note that every directed graph can be expressed as a directed acyclic graph 
%or as known as DAG 
of its strongly connected components. Hence, we can always find a strongly connected component without outgoing edges.
\end{proof}

\begin{thm}
Algorithm \ref{alg:seq-super-stable} computes a maximal sequence of super-stable matchings.
\end{thm}
\begin{proof}
By Lemma \ref{lem:strict} and Lemma \ref{lem:necc}, it is obvious that Algorithm \ref{alg:seq-super-stable} outputs a maximal sequence of super-stable matchings.
\end{proof}

\subsection{Running Time of Algorithm \ref{alg:seq-super-stable}}
\begin{thm}
The running time of Algorithm \ref{alg:seq-super-stable} is $O(mn)$.
\end{thm}
\begin{proof}
Each time we add new edges into $E_d$, we need to compute strongly connected components of $G_d$. Computing strongly connected component of any directed graph $G' = (V',E')$ can be done in $O(E)$ time. Each edge $e$ of $G$ is added to $G_d$ at most once, and $G_d$ is always a subgraph of $G$. Hence, a naive implementation takes $O(m^2)$ on computing strongly connected components of $G_d$. As mentioned in \cite{kunysz2016characterisation}, Pearce \cite{pearce2005some} and Pearce and Kelly \cite{pearce2003online} sketch how to extend their algorithm and that of Marchetti-Spaccamela et al. \cite{marchetti1996maintaining} to compute strongly connected component dynamically. Their algorithm runs in $O(mn)$ if edges can only be added to the graph and not deleted. The edges in $G_d$ can be deleted during the algorithm, but they are deleted only when $E_c$ is perfect on a strongly connected component without outgoing edges. Thus, other strongly connected components are unchanged. Also as mentioned in \cite{kunysz2016characterisation}, the edges remaining in the selected strongly connected component can be treated as they were added anew to the graph. Since the ranks of men increase as we output subsequent super-stable matchings, each edge can be added anew to $G_d$ constant number of times. Thus, the amortized cost of edge insertion remains unchanged. 
The reader can easily check the other part of the algorithm takes at most $O(m)$ time. Hence, the total time is $O(mn)$.
\end{proof}

\subsection{Rotation Poset}
We have shown all rotations $D(\mathcal{M}_G)$ can be found in time $O(mn)$ by Algorithm \ref{alg:seq-super-stable}. It remains to show how to efficiently construct the precedence ration $\prec$ on $D(\mathcal{M}_G)$. Our construction is essentially the same as the construction given in \cite{DBLP:books/daglib/0066875} for the classical stable marriage problem. The only difference here is that one rotation for super-stable matchings can be one or multiple cycles, while one rotation for stable matchings in the classical stable marriage problem is always a cycle.  For the completeness, we briefly sketch it. The reader can find more details in \cite{DBLP:books/daglib/0066875}.

Let $\rho = \{(m_1, w_1), \cdots, (m_{k-1}, w_{k-1})\}$ be a rotation. Each rotation corresponds to the symmetric difference of two consecutive super-stable matchings in a maximal chain of $\mathcal{M}_G$. We say $\rho$ is exposed to a super-stable matching $M$ if $\rho \in M$ and we can eliminate it to obtain another super-stable matching $M' = M \oplus \rho$.

We say that $\rho$ moves $m_i$ down to $w_{i+1}$ and moves $w_i$ up to $m_{i-1}$ for each $i \in [k]$. We also say $\rho$ moves $m_i$ below $w$ if $w_{i-1} \prec_{m_i} w \preceq_{m_i} w_i$ and moves $w_i$ above $m$ if $m_i \preceq_{w_i} m \prec_{w_i} m_{i-1}$.

Let us consider a directed graph $(D(\mathcal{M}_G), E)$, $E$ contains two types of edges:
\begin{itemize}
    \item Type 1: For each pair $(m,w) \in \rho$, if $\rho'$ is the unique rotation that moves $m$ down to $w$, then $(\rho', \rho) \in E$.
    \item Type 2: If $\rho$ moves $m$ below $w$ and $\rho' \neq \rho$ is the unique rotation that moves $w$ above $m$, then $(\rho', \rho) \in E$.
\end{itemize}
The algorithm to construct the set $E$ is simple. For each rotation $\rho$ and each pair $(m_i,w_i) \in \rho$, label $w_i$ in $m_i$'s preference list with a type 1 label of $\rho$, and for each $m$ strictly between $m_i$ and $m_{i-1}$ on $w_i$'s list, label $w_i$ in $m$'s preference list with a type 2 label of $\rho$. Then, traverse each woman $w$ on $m$'s preference list, set $\rho^* = \emptyset$. If $w$ has a type 1 label of $\rho$, add $(\rho^*, \rho)$ into $E$ and set $\rho^* = \rho$. If $w$ has a type 2 label of $\rho$, add $(\rho,\rho^*)$ into $E$. The algorithm takes $O(m)$ time since it only traverse the preference lists once.

It turns out that the transitive closure of $(D(\mathcal{M}_G), E)$ is exactly $(D(\mathcal{M}_G), \prec)$. The reader can find the proof in Irving and Gusfield's book \cite{DBLP:books/daglib/0066875}.

We summarize Section \ref{sec:sequence} with the following theorem.
\begin{thm}
The partial order $(D(\mathcal{M}_G), \prec)$ can be constructed in $O(mn)$.
\end{thm}
\begin{proof}
The construction of $D(\mathcal{M}_G)$ takes $O(mn)$ time by running Algorithm \ref{alg:seq-super-stable}. The precedence relation can be constructed in $O(m)$ time. Hence, the time complexity is $O(mn)$. 
\end{proof}

\section{The Super-stable Matching Polytope}
In this section, we give a polyhedral characterisation of the set of all super-stable matchings and prove that the super-stable matching polytope is integral. The main result is the following theorem.
\begin{thm}\label{thm:susm}
Let $G = (V,E)$ be a stable matching problem with ties where the graph $G$ is bipartite, then the super-stable matching polytope $SUSM(G)$ is described by the following linear system:

\begin{subequations}
\begin{align}
    & \sum_{u \in N(v)} x_{u,v} \leq 1, & & \forall v \in V,\label{eq:susm-a}\\
    & \sum_{i >_u v} x_{u,i} + \sum_{j >_v u} x_{j,v} + x_{u,v} \geq 1, & & \forall (u,v) \in E,\label{eq:susm-b}\\
    & x_{u,v} \geq 0, & & \forall (u,v) \in E \label{eq:susm-c}
\end{align}
\end{subequations}
where $N(v)$ denotes the set of neighbors of $v$ in $G$, and $w>_u v$ means $u$ prefers $w$ to $v$.
\end{thm}

\begin{proof}
Let $x$ be a feasible solution. Define $E^+$ to be the set of edges $(u,v)$ with $x_{u,v} > 0$, and $V^+$ the set of vertices covered by $E^+$. For each $u \in V^+$, let $N^*(u)$ be the maximal elements in $\{i : x_{u,i} > 0\}$. Note that there might be multiple maximal elements that form a tie. 

We first show the following lemma.
\begin{lem}\label{lem:susm-equal}
For each vertex $u$ and each vertex $v \in N^*(u)$, then $u$ is the unique minimal element in $\{j : x_{j,v} > 0\}$ and that $\sum_{j \in N(v)} x_{j,v} = 1$. 
\end{lem}

\begin{proof}
Indeed, (\ref{eq:susm-b}) implies
\begin{equation}\label{eq:super-equal}
    1 \leq \sum_{j >_v u} x_{j,v} + x_{u,v} = \sum_{j \in N(v)} x_{j,v} - \sum_{j <_v u} x_{j,v} - \sum_{\substack{j =_v u;\\j \neq u}} x_{j,v} \leq 1 - \sum_{j <_v u} x_{j,v} - \sum_{\substack{j =_v u;\\j \neq u}} x_{j,v} \leq 1
\end{equation}
Hence we have equality throughout in (\ref{eq:super-equal}). This implies that $x_{j,v} = 0$ for each $\{j : j <_v u\}$ and each $\{j : j =_v u; j \neq u\}$ and that $\sum_{j \in N(v)} x_{j,v} = 1$. Since $x_{j,v} = 0$ for each $\{j : j =_v u; j \neq u\}$, $v$ strictly prefers any other vertices in $\{j : x_{j,v} > 0\}$ over $u$, making $u$ the unique minimal element in $\{j : x_{j,v} > 0\}$.
\end{proof}

We then prove that for any $v$ such that $v \in N^*(u)$ for some $u$, then $u$ is unique. Suppose not, there is a vertex $u' \neq u$ and $v \in N^*(u')$. By Lemma \ref{lem:susm-equal}, $u$ is the unique minimal element in $\{j : x_{j,v} > 0\}$, and $u'$ is the unique minimal element in $\{j : x_{j,v} > 0\}$, contradiction.

Now let $U$ and $W$ be the color classes of $G$. For any $u \in U \cap V^+$, there is at least one unique vertex $w \in N^*(u)$, such that $\sum_{j \in N(w)} x_{j,w} = 1$. Let $F_W(x)$ be the set of these vertices. Formally, $F_W(x) = \{w : w \in N^*(u), u \in U \cap V^+\}$. Then we have $|F_W(x)| \geq |U \cap V^+|$. We also have that

\begin{equation}
    |F_W(x)| = \sum_{w \in F_W(x)} \sum_{j \in N(w)} x_{j,w} = \sum_{j \in U \cap V^+} \sum_{w \in F_W(x)} x_{j,w} \leq \sum_{j \in U \cap V^+} 1 = |U \cap V^+|
\end{equation}
implying that $|F_W(x)| = |U \cap V^+|$. Hence, we conclude that for each $u \in U \cap V^+$, $|N^*(u)| = 1$, which implies that $u$ has an unique maximal element in $\{i : x_{u,i} > 0\}$. Since $|N^*(u)| = 1$, we denote this unique vertex as $x^*(u)$. We then have the following corollary.
\begin{cor}
There is a bijection between $U \cap V^+$ and $F_W(x)$, and for each $u \in U \cap V^+$, $\sum_{i \in N(u)} x_{u,i} = 1$.
\end{cor}

Similarly, we may define $F_U(x) = \{u : u \in N^*(w), w \in W \cap V^+\}$ and we have
\begin{cor}
There is a bijection between $W \cap V^+$ and $F_U(x)$, and for each $w \in W \cap V^+$, $\sum_{j \in N(w)} x_{j,w} = 1$.
\end{cor}

Then we have $|U \cap V^+| = |F_W(x)| \leq |W \cap V^+|$ and $|W \cap V^+| = |F_U(x)| \leq |U \cap V^+|$, implying $|U \cap V^+| = |W \cap V^+| = |F_W(x)| = |F_U(x)|$. Then any $u \in U \cap V^+$ is also in $F_U(x)$, hence, $u$ has an unique minimal element, denoted by $x_*(u)$.

The bijection between $U \cap V^+$ and $F_W(x)$ forms a perfect matching $M$ in $(V^+, E^+)$, i.e. the set of edges $\{(u, x^*(u)) : u \in U \cap V^+\}$. Similarly, the bijection between $W \cap V^+$ and $F_U(x)$ forms another perfect matching $N$, i.e. the set of edges $\{(x^*(w), w) : w \in W \cap V^+\}$. 

Consider the vector $x' = x + \varepsilon\chi^M - \varepsilon\chi^N$, with $\varepsilon$ close enough to $0$ (positive or negative). we will show that $x'$ is also feasible solution of (\ref{eq:susm-a})-(\ref{eq:susm-c}). It is easy to see that $x'$ satisfies (\ref{eq:susm-a}) and (\ref{eq:susm-c}). For each vertex $u \in U \cap V^+$, there is an unique maximal element $x^*(u)$ and $(u, x^*(u)) \in M$ and an unique minimal element $x_*(u)$ and $(u, x_*(u)) \in N$, implying $\sum_{i \in N(u)} x'_{u,i} = \sum_{i \in N(u)} x_{u,i} \leq 1$. To see that $x'$ satisfies (\ref{eq:susm-b}), let $(u,v)$ be an edge in $E^+$ attaining equality in (\ref{eq:susm-b}). The case that $(u,v) \in M$ or $(u,v) \in N$ is trivial. So assume that $(u,v) \notin M$ and $(u,v) \notin N$. The edge $(u, x^*(u)) \in M$ and $x^*(u) >_u v$. There is no other edge in $\{(u, i): i \in N(u)\}$ belongs to $M$. We prove that there is no edge $(j,v)$ in $M$ and $j >_v u$ since if $(j,v) \in M$, $j$ is the minimal element of $v$. Similarly, we can prove that there is exact one edge $(j,v) \in N$ and $j >_v u$. Concluding, $\sum_{i >_u v} x'_{u,i} + \sum_{j >_v u} x'_{j,v} + x'_{u,v} = \sum_{i >_u v} x_{u,i} + \sum_{j >_v u} x_{j,v} + x_{u,v} = 1$. Let $x$ be an extreme point. The feasibility of $x'$ implies that $\chi^M = \chi^N$, that is, $M = N$. So $E^+ = M$ since the maximal element is the same as the minimal element for each vertex, hence, $x = \chi^M$.  
\end{proof}

\subsection{Self-Duality}
Let's consider the linear program:
\begin{maxi!}
    {}{\sum_{(u,v) \in E} x_{u,v}}{\label{lp:primal-susm}}{(\bf{LP})}
\addConstraint{x \in FSUSM(G,P).}
\end{maxi!}

The dual problem with variables $(\alpha, \beta) \in R^V \times R^E$, is given by

\begin{mini!}
    {}{\sum_{v \in V} \alpha_v - \sum_{(u,v) \in E} \beta_{u,v}}{\label{lp:dual-susm}}{(\bf{DLP})}
    \addConstraint{\alpha_u + \alpha_v - \sum_{i <_u v} \beta_{u,i} - \sum_{j <_v u} \beta_{v,j} - \beta_{u,v}}{\geq 1,}{\h\forall (u,v) \in E,}
    \addConstraint{\alpha_v}{ \geq 0,}{\h\forall v \in V,}
    \addConstraint{\beta_{u,v}}{ \geq 0,}{\h\forall (u,v) \in E.}
\end{mini!}

\begin{lem}[\bf{Self-Duality}]
Each $x \in FSUSM(G,P)$ is an optimal solution of $(\bf{LP})$ and $(\alpha, x)$ is an optimal solution of $(\bf{DLP})$, where
\begin{equation}\label{eq:super-nb}
    \alpha_v = \sum_{i \in N(v)} x_{v,i} \h \forall v \in V.
\end{equation}
\end{lem}

\begin{proof}
Let $x \in FSUSM(G,P)$ and let $\alpha$ be defined by (\ref{eq:super-nb}). Let $(u,v) \in E$, we have that

\begin{align*}
    & \alpha_u + \alpha_v - \sum_{i <_u v} x_{u,i} - \sum_{j <_v u} x_{v,j} - x_{u,v}\\
   =& \sum_{i \in N(u)} x_{u,i} + \sum_{j \in N(v)} x_{v,j} - \sum_{i <_u v} x_{u,i} - \sum_{j <_v u} x_{v,j} - x_{u,v}\\
   =& \sum_{i >_u v} x_{u,i} + \sum_{\substack{i =_u v;\\i \neq v}} x_{u,i} + \sum_{j >_v u} x_{v,j} + \sum_{\substack{j =_v u;\\j \neq u}} x_{v,j} + x_{u,v}\\
   \geq& \sum_{i >_u v} x_{u,i} + \sum_{j >_v u} x_{v,j} + x_{u,v}\\
   \geq& 1,
\end{align*}

where the last two inequalities hold since $x$ satisfies (\ref{eq:susm-c}) and (\ref{eq:susm-b}). Hence $(\alpha, x)$ is feasible for $(DLP)$. To see that $x$ and $(\alpha, x)$ are optimal solutions of $(LP)$ and $(DLP)$, respectively. Note that
$$\sum_{v \in V} \alpha_v - \sum_{(u,v) \in E} x_{u,v} = 2 \sum_{(u,v) \in E} x_{u,v} - \sum_{(u,v) \in E} x_{u,v} = \sum_{(u,v) \in E} x_{u,v},$$
so the objective function in $(LP)$ and that in $(DLP)$ are equal. Thus, the weak duality theorem of linear programming implies that $x$ is optimal for $(LP)$ and $(\alpha, x)$ is optimal for $(DLP)$.
\end{proof}

\subsection{Partial Order Preference Lists}
Partial order preference lists are generalisation of preference lists with ties in such a way that the preference list of each man or woman is an arbitrary partial order. It turns out that the linear system (\ref{eq:susm-a})-(\ref{eq:susm-c}) can also describe the set of all super-stable matchings with partial order preference list. 
\begin{proof}
Let $x$ be a feasible solution. Define $E^+$ to be the set of edges $(u,v)$ with $x_{u,v} > 0$, and $V^+$ the set of vertices covered by $E^+$. For each $u \in V^+$, let $N^*(u)$ be the maximal elements in $\{i : x_{u,i} > 0\}$. Note that there might be multiple maximal elements that are incomparable to each other. 

The following lemma is an analogue to Lemma \ref{lem:susm-equal}.
\begin{lem}
For each vertex $u$ and each vertex $v \in N^*(u)$, then $u$ is the unique minimal element in $\{j : x_{j,v} > 0\}$ and that $\sum_{j \in N(v)} x_{j,v} = 1$. 
\end{lem}

\begin{proof}
Indeed, (\ref{eq:susm-b}) implies
\begin{equation}
    1 \leq \sum_{j >_v u} x_{j,v} + x_{u,v} = \sum_{j \in N(v)} x_{j,v} - \sum_{j <_v u} x_{j,v} - \sum_{j \parallel_v u} x_{j,v} \leq 1 - \sum_{j <_v u} x_{j,v} - \sum_{j \parallel_v u} x_{j,v} \leq 1
\end{equation}, where $w \parallel_u v$ means $w$ is incomparable with $v$ in $u$'s preference list.
Hence we have equality throughout in (\ref{eq:super-equal}). This implies that $x_{j,v} = 0$ for each $\{j : j <_v u\}$ and each $\{j : j \parallel_v u\}$ and that $\sum_{j \in N(v)} x_{j,v} = 1$. Since $x_{j,v} = 0$ for each $\{j : j \parallel_v u\}$, $v$ strictly prefers any other vertices in $\{j : x_{j,v} > 0\}$ over $u$, making $u$ the unique minimal element in $\{j : x_{j,v} > 0\}$.
\end{proof}
The rest of the proof is essentially the same as in Theorem \ref{thm:susm}.
\end{proof}

\subsection{The Strongly Stable Matching Polytope}
Kunysz \cite{kunysz2018algorithm} gives a linear system that characterizes the set of all strongly stable matchings and proves this linear system is integral using the duality theory of linear programming. Here, we give an alternate and simpler proof that does not rely on the duality theory and uses only Hall's theorem. 

\begin{thm}[Kunysz, \cite{kunysz2018algorithm}]\label{thm:ssm}
Let $G = (V,E)$ be a stable matching problem with ties where the graph $G$ is bipartite, then the strongly stable matching polytope $SSM(G)$ is described by the following linear system:
\begin{subequations}
\begin{align}
    & \sum_{u \in N(v)} x_{u,v} \leq 1, & & \forall v \in V, \label{eq:ssm-a}\\
    & \sum_{i >_u v} x_{u,i} + \sum_{j >_v u} x_{j,v} + \sum_{k =_u v} x_{u,k} \geq 1, & & \forall (u,v) \in E, \label{eq:ssm-b}\\
    & \sum_{i >_u v} x_{u,i} + \sum_{j >_v u} x_{j,v} + \sum_{k =_v u} x_{k,v} \geq 1, & & \forall (u,v) \in E, \label{eq:ssm-c}\\
    & x_{u,v} \geq 0, & & \forall (u,v) \in E \label{eq:ssm-d}
\end{align}
\end{subequations}
where $N(v)$ denotes the set of neighbors of $v$ in $G$, and $w>_u v$ means $u$ prefers $w$ to $v$.
\end{thm}

We give an alternative proof that does not rely on the duality theory of linear programming.

\begin{proof}
It is easy to verify that the incidence vector of any strongly stable matching satisfies constraints (\ref{eq:ssm-a})-(\ref{eq:ssm-d}). We need to prove each extreme point of the polytope defined by (\ref{eq:ssm-a})-(\ref{eq:ssm-d}) is integral. 

Let $x$ be a feasible solution. Define $E^+$ to be the set of edges $(u,v)$ with $x_{u,v} > 0$, and $V^+$ the set of vertices covered by $E^+$. For each $u \in V^+$, let $N^*(u)$ be the set of maximal elements in $\{i : x_{u,i} > 0\}$. Note that there might be multiple maximal elements that form a tie. Similarly, for each $u \in V^+$, let $N_*(u)$ be the set of minimal elements in $\{i : x_{u,i} > 0\}$.

We first show the following lemma.
\begin{lem}\label{lem:ssm-equal}
For each vertex $u$ and each vertex $v \in N^*(u)$, then $u \in N_*(v)$ and that $\sum_{j \in N(v)} x_{j,v} = 1$, and $\sum_{k=_u v}x_{u,k} \geq \sum_{k=_v u} x_{k,v}$.
\end{lem}

\begin{proof}
Indeed, (\ref{eq:ssm-c}) implies
\begin{equation}\label{eq:strong-equal}
    1 \leq \sum_{j >_v u} x_{j,v} + \sum_{k=_v u}x_{k,v} \leq \sum_{j \in N(v)} x_{j,v} - \sum_{j <_v u} x_{j,v} \leq 1 - \sum_{j <_v u} x_{j,v} \leq 1
\end{equation}
Hence we have equality throughout in (\ref{eq:strong-equal}). This implies that $x_{j,v} = 0$ for each $\{j : j <_v u\}$ and that $\sum_{j \in N(v)} x_{j,v} = 1$. Hence, $u \in N_*(v)$. (\ref{eq:ssm-b}) implies
\begin{equation*}
    1 \leq \sum_{j >_v u} x_{j,v} + \sum_{k=_u v}x_{u,k} \leq 1 - \sum_{k=_v u} x_{k,v} + \sum_{k=_u v}x_{u,k}
\end{equation*}
Hence, $\sum_{k=_u v}x_{u,k} \geq \sum_{k=_v u} x_{k,v}$.
\end{proof}

Now let $U$ and $W$ be the color classes of $G$. Let $E^*(x)$ be the set of edges $\{(u,v) : u \in U \cap V^+, v \in N^*(u)\}$ and $F_W(x)$ be the set of vertices in $W$ and covered by $E^*(x)$. We will show that the subgraph induced by $E^*(x)$ contains a perfect matching $M$.

Suppose not, by Hall's theorem, let $S$ be the unique critical subset of $U \cap V^+$. A subset of $X$ is critical if it is maximally deficient and contains no maximally deficient proper subset. Hence, we have $|N(S)| < |S|$ and there is a matching $M'$ saturating for $N(S)$ in the subgraph induced by $S \cup N(S)$. Fixing the matching $M'$, let $S'$ be the set of vertices in $S$ that is matched in $M'$, so $S' \subsetneq S$. By Lemma \ref{lem:ssm-equal}, we have 
\begin{align}
    \sum_{u \in S'} \sum_{k \in N^*(u)} x_{u,k} &\geq \sum_{v \in N(S)} \sum_{k =_v M'(v)} x_{k,v}\\
    &\geq \sum_{v \in N(S)} \sum_{k \in S} x_{k,v}\\
    &= \sum_{u \in S} \sum_{k \in N^*(u)} x_{u,k}
\end{align}
The first equality follows from Lemma \ref{lem:ssm-equal}. The second inequality follows from that each vertex $v \in N(S)$ is indifferent with all neighbors. The third equality follows by double counting. Hence, the vertices in $S\backslash S'$ are isolated, contradiction. So we also have $|U \cap V^+| = |F_W(x)|$. Again by Lemma \ref{lem:ssm-equal}, we have 

\begin{equation}\label{eq:ssm-equal2}
    \sum_{u \in U \cap V^+} \sum_{k=_u M(u)} x_{u,k} \geq \sum_{v \in F_W(x)} \sum_{k=_v M(v)} x_{k,v} \geq  \sum_{u \in U \cap V^+} \sum_{k=_u M(u)} x_{u,k}
\end{equation}
Hence, we have equality throughout (\ref{eq:ssm-equal2}). So for each $v \in F_W(x)$ and each $k \in N_*(v)$, $v = N^*(k)$. 

Similarly, let $E_*$ be the set of edges $\{(u,v) : v \in W\cap V^+, u \in N^*(v)\}$ and $F_U(x)$ be the set of vertices in $U$ and covered by $E^+$. The subgraph induced by $E_*$ contains a perfect matching $N$. Hence, $|W\cap V^+| = |F_U(x)|$. So $|U\cap V^+| = |W\cap V^+| = |F_U(x)| = |F_W(x)|$. It follows that $E_*$ is exactly the set of edges $\{(u,v) : u \in U \cap V^+, v \in N_*(u)\}$.

Consider the vector $x' = x + \varepsilon\chi^M - \varepsilon\chi^N$, with $\varepsilon$ close enough to $0$ (positive or negative). we will show that $x'$ is also feasible solution of (\ref{eq:ssm-a})-(\ref{eq:ssm-d}). It is easy to see that $x'$ satisfies (\ref{eq:ssm-a}) and (\ref{eq:ssm-d}). For each vertex $u \in U \cap V^+$, there is one edge in $M$ and one edge in $N$ incident to it, implying $\sum_{i \in N(u)} x'_{u,i} = \sum_{i \in N(u)} x_{u,i} \leq 1$. To see that $x'$ satisfies (\ref{eq:ssm-b}), let $(u,v)$ be an edge in $E^+$ attaining equality in (\ref{eq:ssm-b}). First we know that there is one edge in $M$ incident to $u$ such that $M(u) \geq_u v$. Also there is no $M$ edge incident to $v$ and $M(v) >_v u$ since $M(v) \in N_*(v)$. Similarly, there is one edge in $N$ incident to $v$ such that $M(v) \geq_v u$. Then we have two cases. If there is one edge in $N$ incident to $v$ such that $N(v) >_v u$, then there is no edge in $N$ incident to $u$ such that $M(u) \geq_u v$, because otherwise, $u \in N^*(v)$, which contradicts with $N(v) >_v u$. The other case is that there is one edge in $N$ incident to $v$ such that $N(v) =_v u$, then there must be an edge in $N$ incident to $u$ such that $N(u) =_u v$. Hence, $\sum_{i >_u v} x'_{u,i} + \sum_{j >_v u} x'_{j,v} + \sum_{k =_u v} x'_{u,k} = \sum_{i >_u v} x_{u,i} + \sum_{j >_v u} x_{j,v} + \sum_{k =_u v} x_{u,k} = 1$. Similarly, we can prove $x'$ satisfies (\ref{eq:ssm-c}). Let $x$ be an extreme point. The feasibility of $x'$ implies that $\chi^M = \chi^N$, that is, $M = N$. So for each $v \in V^+$, we have $N^*(v) = N_*(v)$ and $\sum_{k \in N^*(v)} x_{k,v} = 1$. Then $x$ is an extreme point of the perfect matching polytope of $(V^+, E^+)$. $x$ must be a perfect matching in $(V^+, E^+)$.
\end{proof}

\bibliographystyle{IEEEtran}
\bibliography{references}

\begin{thebibliography}{10}
\providecommand{\url}[1]{#1}
\csname url@samestyle\endcsname
\providecommand{\newblock}{\relax}
\providecommand{\bibinfo}[2]{#2}
\providecommand{\BIBentrySTDinterwordspacing}{\spaceskip=0pt\relax}
\providecommand{\BIBentryALTinterwordstretchfactor}{4}
\providecommand{\BIBentryALTinterwordspacing}{\spaceskip=\fontdimen2\font plus
\BIBentryALTinterwordstretchfactor\fontdimen3\font minus
  \fontdimen4\font\relax}
\providecommand{\BIBforeignlanguage}[2]{{%
\expandafter\ifx\csname l@#1\endcsname\relax
\typeout{** WARNING: IEEEtran.bst: No hyphenation pattern has been}%
\typeout{** loaded for the language `#1'. Using the pattern for}%
\typeout{** the default language instead.}%
\else
\language=\csname l@#1\endcsname
\fi
#2}}
\providecommand{\BIBdecl}{\relax}
\BIBdecl

\bibitem{spieker1995set}
B.~Spieker, ``The set of super-stable marriages forms a distributive lattice,''
  \emph{Discrete applied mathematics}, vol.~58, no.~1, pp. 79--84, 1995.

\bibitem{manlove2002structure}
D.~F. Manlove, ``The structure of stable marriage with indifference,''
  \emph{Discrete Applied Mathematics}, vol. 122, no. 1-3, pp. 167--181, 2002.

\bibitem{DBLP:books/daglib/0066875}
D.~Gusfield and R.~W. Irving, \emph{The Stable marriage problem - structure and
  algorithms}, ser. Foundations of computing series.\hskip 1em plus 0.5em minus
  0.4em\relax {MIT} Press, 1989.

\bibitem{irving1994stable}
R.~W. Irving, ``Stable marriage and indifference,'' \emph{Discrete Applied
  Mathematics}, vol.~48, no.~3, pp. 261--272, 1994.

\bibitem{kunysz2018algorithm}
A.~Kunysz, ``An algorithm for the maximum weight strongly stable matching
  problem,'' in \emph{29th International Symposium on Algorithms and
  Computation (ISAAC 2018)}.\hskip 1em plus 0.5em minus 0.4em\relax Schloss
  Dagstuhl-Leibniz-Zentrum fuer Informatik, 2018.

\bibitem{scott2005study}
S.~Scott, ``A study of stable marriage problems with ties,'' Ph.D.
  dissertation, University of Glasgow, 2005.

\bibitem{fleiner2007efficient}
T.~Fleiner, R.~W. Irving, and D.~F. Manlove, ``Efficient algorithms for
  generalized stable marriage and roommates problems,'' \emph{Theoretical
  computer science}, vol. 381, no. 1-3, pp. 162--176, 2007.

\bibitem{vate1989linear}
J.~H.~V. Vate, ``Linear programming brings marital bliss,'' \emph{Operations
  Research Letters}, vol.~8, no.~3, pp. 147--153, 1989.

\bibitem{rothblum1992characterization}
U.~G. Rothblum, ``Characterization of stable matchings as extreme points of a
  polytope,'' \emph{Mathematical Programming}, vol.~54, no. 1-3, pp. 57--67,
  1992.

\bibitem{teo1998geometry}
C.-P. Teo and J.~Sethuraman, ``The geometry of fractional stable matchings and
  its applications,'' \emph{Mathematics of Operations Research}, vol.~23,
  no.~4, pp. 874--891, 1998.

\bibitem{kavitha2007strongly}
T.~Kavitha, K.~Mehlhorn, D.~Michail, and K.~E. Paluch, ``Strongly stable
  matchings in time o (nm) and extension to the hospitals-residents problem,''
  \emph{ACM Transactions on Algorithms (TALG)}, vol.~3, no.~2, pp. 15--es,
  2007.

\bibitem{kunysz2016characterisation}
A.~Kunysz, K.~Paluch, and P.~Ghosal, ``Characterisation of strongly stable
  matchings,'' in \emph{Proceedings of the twenty-seventh annual ACM-SIAM
  symposium on Discrete algorithms}.\hskip 1em plus 0.5em minus 0.4em\relax
  SIAM, 2016, pp. 107--119.

\bibitem{iwama1999stable}
K.~Iwama, S.~Miyazaki, Y.~Morita, and D.~Manlove, ``Stable marriage with
  incomplete lists and ties,'' in \emph{International Colloquium on Automata,
  Languages, and Programming}.\hskip 1em plus 0.5em minus 0.4em\relax Springer,
  1999, pp. 443--452.

\bibitem{mcdermid20093}
E.~McDermid, ``A 3/2-approximation algorithm for general stable marriage,'' in
  \emph{International Colloquium on Automata, Languages, and
  Programming}.\hskip 1em plus 0.5em minus 0.4em\relax Springer, 2009, pp.
  689--700.

\bibitem{paluch2014faster}
K.~Paluch, ``Faster and simpler approximation of stable matchings,''
  \emph{Algorithms}, vol.~7, no.~2, pp. 189--202, 2014.

\bibitem{kiraly2013linear}
Z.~Kir{\'a}ly, ``Linear time local approximation algorithm for maximum stable
  marriage,'' \emph{Algorithms}, vol.~6, no.~3, pp. 471--484, 2013.

\bibitem{iwama201425}
K.~Iwama, S.~Miyazaki, and H.~Yanagisawa, ``A 25/17-approximation algorithm for
  the stable marriage problem with one-sided ties,'' \emph{Algorithmica},
  vol.~68, no.~3, pp. 758--775, 2014.

\bibitem{huang2014improved}
C.-C. Huang and T.~Kavitha, ``An improved approximation algorithm for the
  stable marriage problem with one-sided ties,'' in \emph{International
  Conference on Integer Programming and Combinatorial Optimization}.\hskip 1em
  plus 0.5em minus 0.4em\relax Springer, 2014, pp. 297--308.

\bibitem{radnai2014approximation}
A.~Radnai, ``Approximation algorithms for the stable matching problem,''
  \emph{E{\"o}tv{\"o}s Lorand University}, 2014.

\bibitem{dean2015factor}
B.~Dean and R.~Jalasutram, ``Factor revealing lps and stable matching with ties
  and incomplete lists,'' in \emph{Proceedings of the 3rd International
  Workshop on Matching Under Preferences}, 2015, pp. 42--53.

\bibitem{lam20191+}
C.-K. Lam and C.~G. Plaxton, ``A (1+ 1/e)-approximation algorithm for maximum
  stable matching with one-sided ties and incomplete lists,'' in
  \emph{Proceedings of the Thirtieth Annual ACM-SIAM Symposium on Discrete
  Algorithms}.\hskip 1em plus 0.5em minus 0.4em\relax SIAM, 2019, pp.
  2823--2840.

\bibitem{birkhoff1937rings}
G.~Birkhoff \emph{et~al.}, ``Rings of sets,'' \emph{Duke Mathematical Journal},
  vol.~3, no.~3, pp. 443--454, 1937.

\bibitem{pearce2005some}
D.~J. Pearce, ``Some directed graph algorithms and their application to pointer
  analysis,'' Ph.D. dissertation, University of London, 2005.

\bibitem{pearce2003online}
D.~J. Pearce and P.~H. Kelly, ``Online algorithms for topological order and
  strongly connected components,'' Citeseer, Tech. Rep., 2003.

\bibitem{marchetti1996maintaining}
A.~Marchetti-Spaccamela, U.~Nanni, and H.~Rohnert, ``Maintaining a topological
  order under edge insertions,'' \emph{Information Processing Letters},
  vol.~59, no.~1, pp. 53--58, 1996.

\end{thebibliography}

\end{document}